\documentclass[12pt,article,onecolumn]{IEEEtran}

\usepackage{float}
\usepackage{graphicx}
\usepackage{amssymb}
\usepackage{epsfig}
\usepackage{amsmath}
\usepackage{amsthm}
\usepackage{amsfonts}
\usepackage{setspace}
\usepackage{url}
\usepackage{dsfont}
\usepackage{ifpdf}

\newtheorem{thm}{Theorem}
\newtheorem{cor}{Corollary}
\newtheorem{lem}{Lemma}

\makeatletter

\newcommand{\Rmnum}[1]{\expandafter\@slowromancap\romannumeral #1@}
\makeatother

\DeclareMathOperator*{\argmin}{arg\,min}

\newcommand{\abs}[1]{\left\vert#1\right\vert}

\newcommand{\norm}[1]{\left\Vert#1\right\Vert}

\newcommand{\etal}{\textit{et al.} }
\newcommand{\pdf}{\textit{pdf} }
\newcommand{\E}{\textnormal{E} }
\newcommand{\Ball}{\textnormal{Ball} }

\newcommand{\mbf}{\mathbf}
\newcommand{\sbf}{\boldsymbol}

\begin{document}

\title{Non-Random Coding Error Exponent for Lattices}


\author{
Yuval~Domb,~\IEEEmembership{Student Member,~IEEE,}
Meir~Feder,~\IEEEmembership{Fellow,~IEEE}
\thanks{A subset of this work was presented at the IEEE International Symposium on Information Theory (ISIT) 2012.}
}


\maketitle

\begin{abstract}
An upper bound on the error probability of specific lattices, based on their distance-spectrum, is constructed.
The derivation is accomplished using a simple alternative to the Minkowski-Hlawka mean-value theorem of the geometry of numbers.
In many ways, the new bound greatly resembles the Shulman-Feder bound for linear codes.
Based on the new bound, an error-exponent is derived for specific lattice \textit{sequences} (of increasing dimension) over the AWGN channel.
Measuring the sequence's gap to capacity, using the new exponent, is demonstrated.
\end{abstract}

\IEEEpeerreviewmaketitle

\section{Introduction}
For continuous channels, Infinite Constellation (IC) codes are the natural coded-modulation scheme.
The encoding operation is a simple one-to-one mapping from the information messages to the IC codewords.
The decoding operation is equivalent to finding the ``closest''\footnotemark{} IC codeword to the corresponding channel output.
\footnotetext{Closest in the sense of the appropriate distance measure for that channel.}
It has been shown that codes based on linear ICs (a.k.a. Lattices) can achieve optimal error performance \cite{j:deBudaTheUpperError}, \cite{j:ErezAchievingCapacity}.
A widely accepted framework for lattice codes' error analysis is commonly referred to as Poltyrev's setting \cite{j:PoltyrevOnCoding}.
In Poltyrev's setting the code's shaping region, defined as the finite subset of the otherwise infinite set of lattice points, is ignored, and the lattice structure is analyzed for its coding (soft packing) properties only.
Consequently, the usual rate variable $R$ is infinite and replaced by the Normalized Log Density (NLD), $\delta$.
The lattice analogous to Gallager's random-coding error-exponent \cite{b:GallagerInformationTheory}, over a random linear codes ensemble, is Poltyrev's error-exponent over a set of lattices of constant density.

Both Gallager and Poltyrev's error-exponents are asymptotic upper bounds for the exponential behavior of the average error probability over their respective ensembles.
Since it is an average property, examining a specific code from the ensemble using these exponents is not possible.
Various upper error bounding techniques and bounds for specific codes and code families have been constructed for linear codes over discrete channels \cite{j:ShamaiVariationsOnThe}, while only a few have been devised for specific lattices over particular continuous channels \cite{j:HerzbergTechniquesofBounding}.

The Shulman-Feder Bound (SFB) \cite{j:ShulmaRandomCoding} is a simple, yet useful upper bound for the error probability of a specific linear code over a discrete-memoryless channel.
In its exponential form it states that the average error probability of a $q$-ary code $\mathcal{C}$ is upper bounded by
\begin{align}
&P_e(\mathcal{C}) \leq e^{-nE_r\left(R+\frac{\log\alpha}{n}\right)} \\
&\alpha = \max_{\tau} \frac{\mathcal{N}_{\mathcal{C}}(\tau)}{\mathcal{N}_r(\tau)} \frac{e^{nR}}{e^{nR}-1}
\end{align}
where $n$, $R$, and $E_r(R)$ are the dimension, rate, and random-coding exponent respectively, and $\mathcal{N}_{\mathcal{C}}(\tau)$ and $\mathcal{N}_r(\tau)$ are the number of codewords of type $\tau$ for $\mathcal{C}$ and for an average random code, respectively (i.e. distance-spectrum).
The SFB and its extensions have lead to significant results in coding theory.
Amongst those is the error analysis for Maximum Likelihood (ML) decoding of LDPC codes \cite{j:MillerBoundsOnThe}.
The main motivation of this paper is to find the SFB analogue for lattices.
As such it should be an expression that upper bounds the error probability of a specific lattice, depend on the lattice's distance spectrum, and resemble the lattice random-coding bound.

The main result of this paper is a simple upper bounding technique for the error probability of a specific lattice code or code family, as a function of its distance-spectrum.
The bound is constructed by replacing the well-known Minkowski-Hlawka theorem \cite{b:LekkerkerkerGeometryOfNumbers} with a non-random alternative.
An interesting outcome of the main result is an error-exponent for specific lattice \textit{sequences}.
A secondary result of this paper is a tight distance-spectrum based, upper bound for the error probability of a specific lattice of finite dimension.

The paper is organized as follows: Sections \Rmnum{2} and \Rmnum{3} present the derivation of the Minkowski-Hlawka non-random alternatives, section \Rmnum{4} outlines a well-known general ML decoding upper bound, section \Rmnum{5} applies the new techniques to the general bound of section \Rmnum{4}, and section \Rmnum{6} presents a new error-exponent for specific lattice \textit{sequences} over the AWGN channel.

\section{Deterministic Minkowski-Hlawka-Siegel}
Recall that a lattice $\Lambda$ is a discrete $n$-dimensional subgroup of the Euclidean space $\mathds{R}^n$ that is an Abelian group under addition.
A generating matrix $G$ of $\Lambda$ is an $n\times n$ matrix with real valued coefficients constructed by concatenation of a properly chosen set of $n$ linearly independent vectors of $\Lambda$.
The generating matrix $G$ defines the lattice $\Lambda$ by $\Lambda=\{\sbf{\lambda}: \sbf{\lambda}=G\mbf{u}, \mbf{u}\in\mathds{Z}^n\}$.
A fundamental parallelepiped of $\Lambda$, associated with $G$ is the set of all points $p=\sum_{i=1}^n u_i g_i$ where $0\leq u_i<1$ and $\{g_i\}_{i=1}^n$ are the basis vectors of $G$.
The lattice determinant, defined as $\det{\Lambda}\equiv\abs{\deg{G}}$, is also the volume of the fundamental parallelepiped.
Denote by $\beta$ and $\delta$ the density and NLD of $\Lambda$ respectively; thus $\beta = e^{n\delta} = (\det{\Lambda})^{-1}$.

The lattice-dual of the random linear codes ensemble, in finite-alphabet codes, is a set of lattices originally defined by Siegel \cite{j:SiegelAMeanValue,j:MacbeathAModifiedForm}, for use in proving what he called the Mean-Value Theorem (MVT).
This theorem, often referred to as the Minkowski-Hlawka-Siegel (MHS) theorem, is a central constituent in upper error bounds on lattices.
The theorem states that for any dimension $n\geq 2$, and any bounded Riemann-integrable function $g(\sbf{\lambda})$ there exists a lattice $\Lambda$ of density $\beta$ for which
\begin{equation}
\label{eq_mhs}
\sum_{\sbf{\lambda}\in\Lambda\setminus\{0\}} g(\sbf{\lambda})
\leq \int_{\mathds{R}^n} g\left(\frac{\mbf{x}}{e^{\delta}}\right) d\mbf{x}
= \beta \int_{\mathds{R}^n} g(\mbf{x}) d\mbf{x}.
\end{equation}
Siegel proved the theorem by averaging over a fundamental set\footnotemark{} of all $n$-dimensional lattices of unit density.
The disadvantages of Siegel's theorem are similar to the disadvantages of the random-coding theorem.
Since the theorem is an average property of the ensemble, it can be argued that there exists at least a single specific lattice, from the ensemble, that obeys it; though finding that lattice cannot be aided by the theorem.
Neither can the theorem aid in analysis of any specific lattice.
Alternatives to \eqref{eq_mhs}, constructed for specific lattices, based on their distance-spectrum, are introduced later in this section.
\footnotetext{Let $\Upsilon$ denote the multiplicative group of all non-singular $n\times n$ matrices with determinant $1$ and let $\Phi$ denote the subgroup of integral matrices in $\Upsilon$.
Siegel's fundamental set is defined as the set of lattices whose generating matrices form a fundamental domain of $\Upsilon$ with regards to right multiplication by $\Phi$ (see section 19.3 of \cite{b:LekkerkerkerGeometryOfNumbers}).}

We begin with a few definitions, before stating our central lemma.
The lattice $\Lambda_0$ always refers to a specific known $n$-dimensional lattice of density $\beta$, rather than $\Lambda$ which refers to some unknown, yet existing $n$-dimensional lattice.
The lattice $\widetilde{\Lambda}_0$ is the normalized version of $\Lambda_0$ (i.e. $\det(\widetilde{\Lambda}_0) = 1$).
Define the distance series of $\Lambda_0$ as the ordered series of its unique norms $\{\lambda_j\}_{j=0}^\infty$, such that $\lambda_1$ is its minimal norm and $\lambda_0\triangleq 0$.
$\{\widetilde{\lambda}_j\}_{j=1}^\infty$ is defined for $\widetilde{\Lambda}_0$ respectively.
The normalized continuous distance-spectrum of $\Lambda_0$ is defined as
\begin{equation}
\label{eq_N}
N(x) = \sum_{j=1}^\infty \mathcal{N}_j \delta(x-\widetilde{\lambda}_j)
\end{equation}
where $\{\mathcal{N}_j\}_{j=1}^\infty$ is the ordinary distance-spectrum of $\Lambda_0$, and $\delta(\cdot)$ is the Dirac delta function.
Let $\Gamma$ denote the group\footnotemark{} of all orthogonal $n\times n$ matrices with determinant $+1$ and let $\mu(\gamma)$ denote its normalized measure so that $\int_{\Gamma} d\mu(\gamma) = 1$.
The notation $\gamma\Lambda_0$ is used to describe the lattice generated by $\gamma G$, where $G$ is a generating matrix of the lattice $\Lambda_0$.
\footnotetext{This group, consisting only of rotation matrices, is usually called the special orthogonal group.}

Our central lemma essentially expresses Siegel's mean-value theorem for a degenerate ensemble consisting of a specific known lattice $\Lambda_0$ and all its possible rotations around the origin.
\begin{lem}
\label{lem_mean_value_1}
Let $\Lambda_0$ be a specific $n$-dimensional lattice with NLD $\delta$, and $g(\sbf{\lambda})$ be a Riemann-integrable function, then there exists an orthogonal rotation $\gamma$ such that
\begin{equation}
\label{eq_mean_value_1}
\sum_{\sbf{\lambda}\in\gamma\Lambda_0\setminus\{0\}} g(\sbf{\lambda})
\leq \int_{\mathds{R}^n} \mathfrak{N}(\norm{\mbf{x}}) g\left(\frac{\mbf{x}}{e^{\delta}}\right) d\mbf{x}
\end{equation}
with
\begin{equation}
\label{eq_N_}
\mathfrak{N}(x)
\triangleq \left\{
\begin{array}{lr}
\frac{N(x)}{nV_n x^{n-1}} & : x > 0 \\
0 & : x \leq 0
\end{array}
\right.
\end{equation}
where $V_n$ is the volume of an $n$-dimensional unit sphere, and $\norm{\cdot}$ denotes the Euclidean norm.
\end{lem}
\begin{proof}
Let $\Theta$ denote the subspace of all points $\sbf{\theta}$ in the $n$-dimensional space with $\norm{\sbf{\theta}}=1$, so that $\Theta$ is the surface of the unit sphere. Let $\mu(\sbf{\theta})$ denote the ordinary solid-angle measure on this surface, normalized so that $\int_{\Theta} d\mu(\sbf{\theta}) = 1$.
We continue with the following set of equalities
\begin{align}
\label{eq_mean_value_1_proof}
\int_\Gamma \sum_{\sbf{\lambda}\in\gamma\Lambda_0\setminus\{0\}} &g(\sbf{\lambda}) d\mu(\gamma) \nonumber \\
&= \sum_{\sbf{\lambda}\in\Lambda_0\setminus\{0\}} \int_\Gamma g(\gamma\sbf{\lambda}) d\mu(\gamma) \nonumber \\
&= \sum_{\mbf{\widetilde{\lambda}}\in\widetilde{\Lambda}_0\setminus\{0\}} \int_\Gamma g\left(\frac{\gamma\sbf{\widetilde{\lambda}}}{e^{\delta}}\right) d\mu(\gamma) \nonumber \\
&= \sum_{\mbf{\widetilde{\lambda}}\in\widetilde{\Lambda}_0\setminus\{0\}} \int_\Theta g\left(\frac{\norm{\sbf{\widetilde{\lambda}}}\sbf{\theta}}{e^{\delta}}\right) d\mu(\sbf{\theta}) \nonumber \\
&= \int_{0^+}^\infty N(R) \int_\Theta g\left(\frac{R\sbf{\theta}}{e^{\delta}}\right) d\mu(\sbf{\theta}) dR \nonumber \\
&= \int_{0^+}^\infty \int_\Theta \frac{N(R)}{nV_nR^{n-1}} \cdot g\left(\frac{R\sbf{\theta}}{e^{\delta}}\right) d\mu(\sbf{\theta}) dV_nR^n \nonumber \\
&= \int_{\mathds{R}^n\setminus\{0\}} \frac{N(\norm{\mbf{x}})}{nV_n\norm{\mbf{x}}^{n-1}} \cdot g\left(\frac{\mbf{x}}{e^{\delta}}\right) d\mbf{x} \nonumber \\
&= \int_{\mathds{R}^n} \mathfrak{N}(\norm{\mbf{x}}) g\left(\frac{\mbf{x}}{e^{\delta}}\right) d\mbf{x}.
\end{align}
where the third equality follows from the definition of $\Gamma$ and $\Theta$ and the measures $\mu(\gamma)$ and $\mu(\sbf{\theta})$, the fourth equality is due to the circular-symmetry of the integrand, and the sixth equality is a transformation from generalized spherical polar coordinates to the cartesian system (see Lemma 2 of \cite{j:MacbeathAModifiedForm}).

Finally there exists at least one rotation $\gamma\in\Gamma$ for which the sum over $\gamma\Lambda_0$ is upper bounded by the average.
\end{proof}
The corollary presented below is a restricted version of lemma \ref{lem_mean_value_1} constrained to the case where the function $g(\sbf{\lambda})$ is circularly-symmetric, (i.e. $g(\sbf{\lambda})=g(\norm{\sbf{\lambda}})$).
To simplify the presentation, it is implicitly assumed that $g(\sbf{\lambda})$ is circularly-symmetric for the remainder of this paper.
It should be noted that all results presented hereafter apply also to a non-symmetric $g(\sbf{\lambda})$ with an appropriately selected rotation $\gamma$ of $\Lambda_0$.
\begin{cor}
\label{lem_mean_value_2}
Let $\Lambda_0$ be a specific $n$-dimensional lattice with NLD $\delta$, and $g(\sbf{\lambda})$ be a circularly-symmetric Riemann-integrable function, then
\begin{equation}
\label{eq_mean_value_2}
\sum_{\sbf{\lambda}\in\Lambda_0\setminus\{0\}} g(\sbf{\lambda})
= \int_{\mathds{R}^n} \mathfrak{N}(\norm{\mbf{x}}) g\left(\frac{\mbf{x}}{e^{\delta}}\right) d\mbf{x}
\end{equation}
with $\mathfrak{N}(x)$ as defined in lemma \ref{lem_mean_value_1}.
\end{cor}
\begin{proof}
When $g(\sbf{\lambda})$ is circularly-symmetric,
\begin{align}
\int_\Gamma \sum_{\sbf{\lambda}\in\gamma\Lambda_0\setminus\{0\}} g(\sbf{\lambda}) d\mu(\gamma)
&= \sum_{\sbf{\lambda}\in\Lambda_0\setminus\{0\}} \int_\Gamma g(\gamma\sbf{\lambda}) d\mu(\gamma) \nonumber \\
&= \sum_{\sbf{\lambda}\in\Lambda_0\setminus\{0\}} g(\sbf{\lambda}) \int_\Gamma d\mu(\gamma) \nonumber \\
&= \sum_{\sbf{\lambda}\in\Lambda_0\setminus\{0\}} g(\sbf{\lambda})
\end{align}
\end{proof}

The right-hand side of \eqref{eq_mean_value_2} can be trivially upper bounded by replacing $\mathfrak{N}(x)$ with a suitably chosen function $\alpha(x)$, so that
\begin{equation}
\label{eq_alphai}
\int_{\mathds{R}^n} \mathfrak{N}(\norm{\mbf{x}}) g\left(\frac{\mbf{x}}{e^{\delta}}\right) d\mbf{x}
\leq \int_{\mathds{R}^n} \alpha(\norm{\mbf{x}}) g\left(\frac{\mbf{x}}{e^{\delta}}\right) d\mbf{x}.
\end{equation}
Provided the substitution, it is possible to define the following upper bounds:
\begin{thm}[Deterministic Minkowski-Hlawka-Siegel (DMHS)]
\label{thm_nrmh1}
Let $\Lambda_0$ be a specific $n$-dimensional lattice of density $\beta$, $g(\sbf{\lambda})$ be a bounded Riemann-integrable circularly-symmetric function, and $\alpha(x)$ be defined such that \eqref{eq_alphai} is satisfied, then
\begin{equation}
\label{eq_nrmh1}
\sum_{\sbf{\lambda}\in\Lambda_0\setminus\{0\}} g(\sbf{\lambda})
\leq \beta \left[ \max_{x\leq e^{\delta}\lambda_{\textnormal{max}}} \alpha(x) \right] \int_{\mathds{R}^n} g(\mbf{x}) d\mbf{x}
\end{equation}
where $\lambda_{\textnormal{max}}$ is the maximal $\norm{\mbf{x}}$ for which $g(\mbf{x})\neq 0$.
\end{thm}
\begin{proof}
Substitute a specific $\alpha(x)$ for $\mathfrak{N}(x)$ in \eqref{eq_mean_value_2}, and upper bound by taking the maximum value of $\alpha(x)$ over the integrated region, outside the integral.
\end{proof}
\begin{thm}[extended DMHS (eDMHS)]
\label{thm_nrmh2}
Let $\Lambda_0$ be a specific $n$-dimensional lattice with density $\beta$, $g(\sbf{\lambda})$ be a bounded Riemann-integrable circularly-symmetric function, $\alpha(x)$ be defined such that \eqref{eq_alphai} is satisfied, and $M$ be an positive integer then
\begin{align}
\label{eq_nrmh2}
\sum_{\sbf{\lambda}\in\Lambda_0/\{0\}} g(\sbf{\lambda})
\leq \beta \min_{\{R_j\}_{j=1}^M} \left( \sum_{j=1}^M \max_{\mathcal{R}_j} \alpha(x) \int_{\sbf{\mathcal{R}}_j} g(\mbf{x}) d\mbf{x} \right)
\end{align}
with
\begin{align}
&\mathcal{R}_j = \{x: x\geq 0,  e^{\delta}R_{j-1}<x\leq e^{\delta}R_j\} \\
&\sbf{\mathcal{R}}_j = \{\mbf{x}: \mbf{x}\in\mathds{R}^n, R_{j-1}<\norm{\mbf{x}}\leq R_j\}
\end{align}
where $\{R_j\}_{j=1}^M$ is an ordered set of real numbers with $R_0\triangleq 0$ and $R_M=\lambda_{\textnormal{max}}$, where $\lambda_{\textnormal{max}}$ is the maximal $\norm{\mbf{x}}$ for which $g(\mbf{x})\neq 0$.
\end{thm}
\begin{proof}
Substitute a specific $\alpha(x)$ for $\mathfrak{N}(x)$ in \eqref{eq_mean_value_2} and break up the integral over a non-overlapping set of spherical shells whose union equals $\mathds{R}^n$. Upper bound each shell integral by taking the maximum value of $\alpha(x)$ over it, outside the integral. Finally, the set of shells, or rather shell-defining radii is optimized such that the bound is tightest.
\end{proof}
The eDMHS may be viewed as a generalization of DMHS since for $M=1$ it defaults to it.
In addition, when $M\rightarrow\infty$ the eDMHS tends to the original integral.
Clearly, both bounds are sensitive to choice of $\alpha(x)$, though one should note, that for the same choice of $\alpha(x)$ the second bound is always tighter.

The bounds shown above are general up to choice of $\alpha(x)$, and clarify the motivation for the substitution in \eqref{eq_alphai}.
Careful construction of $\alpha(x)$ along with the selected bound can provide a tradeoff between tightness and complexity.
The next section presents a few simple methods for construction of the function $\alpha(x)$, and their consequences.

\section{Construction of $\alpha(x)$}
Maximization of the right-hand-side of \eqref{eq_mean_value_2}, by taking out the maximum value of $\mathfrak{N}(x)$ outside the integral, is not well-defined.
This is since $\mathfrak{N}(x)$ is an impulse train.
The motivation of this section is to find a replacement for $\mathfrak{N}(x)$ that enables using this maximization technique whilst retaining a meaningful bound.

Let's assume that $g(\sbf{\lambda})$ is monotonically non-increasing\footnotemark{} in $\norm{\sbf{\lambda}}$, and define $N'(x)$ to be a smoothed version of the normalized continuous distance-spectrum, selected such that it satisfies
\footnotetext{This restriction holds for many continuous noise channels, such as AWGN. For other channels, it is possible to define similar heuristics.}
\begin{equation}
\label{eq_alpha1_2}
\int_{0^+}^r N(R) dR
\leq \int_{0^+}^r N'(R) dR
\qquad \forall r \in (0,e^\delta\lambda_{\textnormal{max}}].
\end{equation}
Given the above, $\alpha(x)$ can be defined by expanding \eqref{eq_alphai} as follows:
\begin{align}
\int_{\mathds{R}^n} \mathfrak{N}(\norm{\mbf{x}}) g\left(\frac{\mbf{x}}{e^{\delta}}\right) d\mbf{x}
&= \int_{0^+}^{e^\delta\lambda_{\textnormal{max}}} N(R) g\left(\frac{R}{e^\delta}\right) dR \nonumber \\
&\leq \int_{0^+}^{e^\delta\lambda_{\textnormal{max}}} N'(R) g\left(\frac{R}{e^\delta}\right) dR \nonumber \\
&= \int_{\mathds{R}^n} \alpha(\norm{\mbf{x}}) g\left(\frac{\mbf{x}}{e^{\delta}}\right) d\mbf{x}
\end{align}
with
\begin{equation}
\alpha(x)
\triangleq \left\{
\begin{array}{lr}
\frac{N'(x)}{nV_n x^{n-1}} & : x > 0 \\
0 & : x \leq 0
\end{array}
\right.
\end{equation}
where the equalities follow methods used in \eqref{eq_mean_value_1_proof} together with $g(\sbf{\lambda})$'s circular-symmetry, and the inequality follows from $g(\sbf{\lambda})$ being monotonically non-increasing together with $N'(x)$ obeying \eqref{eq_alpha1_2}.

Define $\{\alpha_i(x)\}_{i\in I}$ to be the set of all functions $\alpha(x)$ such that \eqref{eq_alpha1_2} is satisfied.
Specifically let's define the following two functions:
\begin{itemize}
\item The first function $\alpha^{\textnormal{rng}}(x)$ is defined to be piecewise constant over shells defined by consecutive radii from the normalized distance series $\{\widetilde{\lambda}_j\}_{j=0}^\infty$, (i.e. a shell is defined as $\mathcal{S}_j = \{x: \widetilde{\lambda}_{j-1}<x\leq\widetilde{\lambda}_j\}$). The constant value for shell $j$ is selected such that the spectral mass $\mathcal{N}_j$ is spread evenly over the shell. Formally, this can be expressed as
    \begin{equation}
    \left[\alpha^{\textnormal{rng}}(x)\right]_{x\in\mathcal{S}_j} = \frac{\mathcal{N}_j}{V_n\left(\widetilde{\lambda}_j^n-\widetilde{\lambda}_{j-1}^n\right)}.
    \end{equation}
\item The second function $\alpha^{\textnormal{opt}}(x)$ is selected to be $\alpha_j(x)$ such that the bound \eqref{eq_nrmh1} is tightest; thus by definition
    \begin{equation}
    \label{eq_alpha1_3}
    j = \argmin_{i \in I} \max_{x\leq  e^{\delta}\lambda_{\textnormal{max}}} \alpha_i(x).
    \end{equation}
    As it turns out $\alpha^{\textnormal{opt}}(x)$ is also piecewise constant over shells defined by consecutive radii from the normalized distance series, and can be obtained as the solution to a linear program presented in appendix \ref{app_linear_program_alpha1}.

    One subtlety, that can go by unnoticed about $\alpha^{\textnormal{opt}}(x)$, is its dependence on $\lambda_{\textnormal{max}}$.
    Careful examination reveals that $\alpha^{\textnormal{opt}}(x)$ is constructed from those spectral elements whose corresponding distances are less than or equal to $\lambda_{\textnormal{max}}$.
    This is of relevance when optimization of $\lambda_{\textnormal{max}}$ is dependent on $\alpha^{\textnormal{opt}}(x)$, as is the case in some of the bounds discussed hereafter.
    This technical subtlety can be overcome by construction of a suboptimal version of $\alpha^{\textnormal{opt}}(x)$ consisting of more spectral elements than necessary.
    In many cases both the suboptimal and optimal versions coincide.
    In the remainder of this paper, this technicality is ignored and left for the reader's consideration.
\end{itemize}

Figure \ref{fig_z2_alpha} illustrates $\mathfrak{N}(x)$, $\alpha^{\textnormal{rng}}(x)$, and $\alpha^{\textnormal{opt}}(x)$ for the rectangular lattice $\mathds{Z}^2$.

\begin{figure}[htp]
\center{\includegraphics[width=0.5\textwidth]{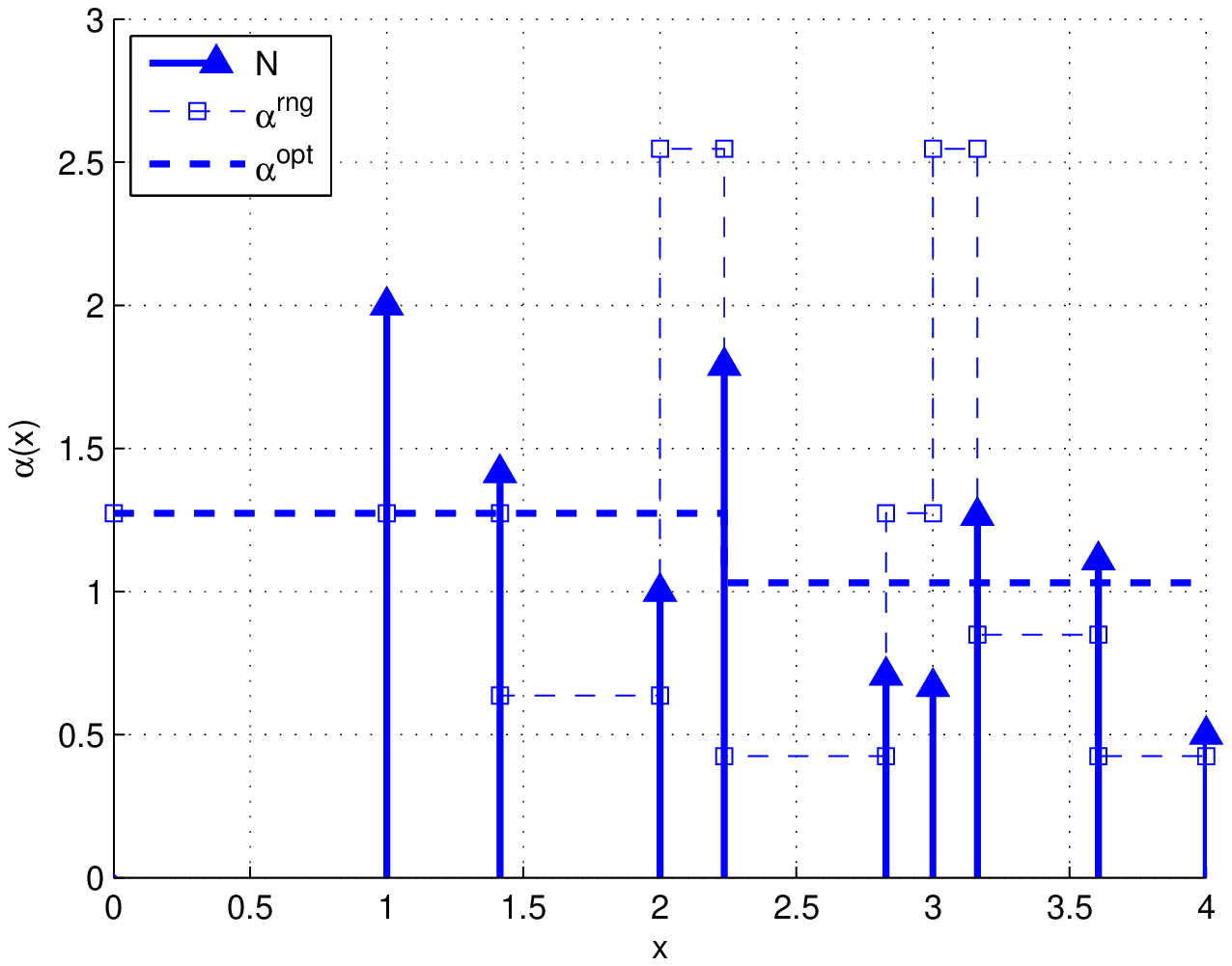}}
\caption{\label{fig_z2_alpha} $\mathfrak{N}(x)$, $\alpha^{\textnormal{rng}}(x)$, and $\alpha^{\textnormal{opt}}(x)$ for the rectangular lattice $\mathds{Z}^2$.}
\end{figure}

\section{A General ML Decoding Upper Bound}
Many tight ML upper bounds originate from a general bounding technique, developed by Gallager \cite{b:GallagerLowDensity}. Gallager's technique has been utilized extensively in literature \cite{j:ShamaiVariationsOnThe,j:YousefiANewUpper,j:TwittoTightenedUpperBounds}.
Similar forms of the general bound, displayed hereafter, have been previously presented in literature \cite{j:HerzbergTechniquesofBounding,j:IngberFiniteDimensional}.
Playing a central role in our analysis, we present it as a theorem.

Before proceeding with the theorem, let us define an Additive Circularly-Symmetric Noise (ACSN) channel as an additive continuous noise channel, whose noise is isotropically distributed and is a non-increasing function of its norm.
\begin{thm}[General ML Upper Bound]
\label{thm_bound_general}
Let $\Lambda$ be an $n$-dimensional lattice, and $f_{\norm{\mbf{z}}}(\rho)$ the \pdf of an ACSN channel's noise vector's norm, then the error probability of an ML decoder is upper bounded by
\begin{align}
\label{eq_bound_general}
P_e(\Lambda)
\leq \min_{r} & \left(\sum_{\sbf{\lambda}\in\Lambda\setminus\{0\}} \int_0^r f_{\norm{\mbf{z}}}(\rho) P_2(\sbf{\lambda},\rho) d\rho \right. \\ \nonumber
& \qquad + \left. \int_r^\infty f_{\norm{\mbf{z}}}(\rho) d\rho \right)
\end{align}
where the pairwise error probability conditioned on $\norm{\mbf{z}}=\rho$ is defined as
\begin{equation}
\label{eq_p2}
P_2(\sbf{\lambda},\rho) = \Pr(\sbf{\lambda}\in \Ball(\mbf{z},\norm{\mbf{z}})|\norm{\mbf{z}}=\rho)
\end{equation}
where $\Ball(\mbf{z},\norm{\mbf{z}})$ is an $n$-dimensional ball of radius $\norm{\mbf{z}}$ centered around $\mbf{z}$.
\end{thm}
\begin{proof}
See appendix \ref{app_bound_general}.
\end{proof}
We call the first term of \eqref{eq_bound_general} the Union Bound Term (UBT) and the second term the Sphere Bound Term (SBT) for obvious reasons.

In general \eqref{eq_p2} is difficult to quantify.
One method to overcome this, which is limited for analysis of specific lattices, is by averaging over the MHS ensemble.
New methods for upper bounding \eqref{eq_bound_general} for specific lattices are presented in the next section.

\section{Applications of the General ML Bound}
In this section, the UBT of the ML decoding upper bound \eqref{eq_bound_general} is further bounded using different bounding methods. The resulting applications vary in purpose, simplicity, and exhibit different performance. We present the applications and discuss their differences.

\subsection{MHS}
Application of the MHS theorem from \eqref{eq_mhs}, leads to the random-coding error bound on lattices.
Since it is based on the MHS ensemble average, the random-coding bound proves the existence of a lattice bounded by it, but does not aid in finding such lattice; neither does it provide tools for examining specific lattices.
\begin{thm}[MHS Bound, Theorem 5 of \cite{j:IngberFiniteDimensional}]
\label{thm_bound_mhs}
Let $f_{\norm{\mbf{z}}}(\rho)$ be the \pdf of an ACSN channel's noise vector's norm, then there exists an $n$-dimensional lattice $\Lambda$ of density $\beta$  for which the error probability of an ML decoder is upper bounded by
\begin{equation}
\label{eq_bound_mhs}
P_e(\Lambda)
\leq \beta V_n \int_0^{r^*} f_{\norm{\mbf{z}}}(\rho) \rho^n d\rho + \int_{r^*}^\infty f_{\norm{\mbf{z}}}(\rho) d\rho
\end{equation}
with
\begin{equation}
r^* = (\beta V_n)^{-1/n}.
\end{equation}
\end{thm}
\begin{proof}
Set $g(\sbf{\lambda})$ as
\begin{equation}
\label{eq_bound_mhs_g}
g(\sbf{\lambda})
= \int_0^r f_{\norm{\mbf{z}}}(\rho) P_2(\sbf{\lambda},\rho) d\rho,
\end{equation}
noting that it is a bounded function of $\sbf{\lambda}$ and continue to bound the UBT from \eqref{eq_bound_general} using \eqref{eq_mhs}.
The remainder of the proof is presented in appendix \ref{app_bound_mhs}.
\end{proof}

\subsection{DMHS}
Application of the DMHS theorem \eqref{eq_nrmh1} using $\alpha(x)=\alpha^{\textnormal{opt}}(x)$ provides a tool for examining specific lattices. The resulting bound is essentially identical to the MHS bound, excluding a scalar multiplier of the UBT. It is noteworthy that this is the best $\alpha(x)$-based bound of this form, since $\alpha^{\textnormal{opt}}(x)$ is optimized with regards to DMHS.
\begin{thm}[DMHS Bound]
\label{thm_bound_nrmh1}
Let a specific $n$-dimensional lattice $\Lambda_0$ of density $\beta$ be transmitted over an ACSN channel with $f_{\norm{\mbf{z}}}(\rho)$ the \pdf of its noise vector's norm, then the error probability of an ML decoder is upper bounded by
\begin{equation}
\label{eq_bound_nrmh1}
P_e(\Lambda_0)
\leq \min_r \left( \alpha \beta V_n \int_0^{r} f_{\norm{\mbf{z}}}(\rho) \rho^n d\rho + \int_{r}^\infty f_{\norm{\mbf{z}}}(\rho) d\rho \right)
\end{equation}
with
\begin{equation}
\label{eq_bound_nrmh1_alpha}
\alpha = \max_{x\leq  e^{\delta}\cdot2r} \alpha^{\textnormal{opt}}(x)
\end{equation}
where $\alpha^{\textnormal{opt}}(x)$ is as defined by \eqref{eq_alpha1_3}.
\end{thm}
\begin{proof}
Set $g(\sbf{\lambda})$ as in \eqref{eq_bound_mhs_g}, noting that it is bounded by $\lambda_{\textnormal{max}}=2r$. The remainder is identical to the proof of theorem \ref{thm_bound_mhs} replacing $\beta$ with $\alpha\beta$.
\end{proof}

Optimization of $r$ can be performed in the following manner.
Since $\alpha^{opt}(x)$ is a monotonically non-increasing function of $x$, optimization of $r$ is possible using an iterative numerical algorithm.
In the first iteration, set $r = (\beta V_n)^{-1/n}$ and calculate $\alpha$ according to \eqref{eq_bound_nrmh1_alpha}.
In each additional iteration, set $r = (\alpha\beta V_n)^{-1/n}$ and recalculate $\alpha$. The algorithm is terminated on the first iteration when $\alpha$ is unchanged.

\subsection{eDMHS}
Rather than maximizing the UBT using a single scalar factor (as was done in the DMHS), the eDMHS splits up the UBT integral to several regions with boundaries defined by $\{\lambda_j\}_{j=0}^\infty$.
Maximization of each resulting region by its own scalar, results in much tighter, yet more complex bound.
This is typically preferred for error bounding in finite dimension lattices.
For asymptotical analysis, the eDMHS would typically require an increasing (with the dimension) number of spectral elements, while the DMHS would still only require one, making it the favorable.
We choose $\alpha(x)=\alpha^{\textnormal{opt}}(x)$, rather than optimizing $\alpha(x)$ for this case.
The reason is that although this bound is tighter for the finite dimension case, it is considerably more complex than a competing bound (for the finite dimension case) presented in the next subsection.
The main motivation for showing this bound is as a generalization of DMHS.
\begin{thm}[eDMHS Bound]
\label{thm_bound_nrmh2}
Let a specific $n$-dimensional lattice $\Lambda_0$ of density $\beta$ be transmitted over an ACSN channel with $f_{\norm{\mbf{z}}}(\rho)$ the \pdf of its noise vector's norm, then the error probability of an ML decoder is upper bounded by
\begin{align}
\label{eq_bound_nrmh2}
P_e(\Lambda_0)
\leq \min_r & \left( \beta \sum_{j=1}^M \alpha_j \int_{\lambda_j/2}^{r} f_{\norm{\mbf{z}}}(\rho) h_j(\rho) d\rho \right. \nonumber \\
& \qquad + \left. \int_{r}^\infty f_{\norm{\mbf{z}}}(\rho) d\rho \right)
\end{align}
with
\begin{align}
&\alpha_j = \max_{\mathcal{S}_j} \alpha^{\textnormal{opt}}(x) = \alpha^{\textnormal{opt}}(e^{\delta}\lambda_j) \\
&h_j(\rho) = \int_{\sbf{\mathcal{S}}_j} \sigma\{\mbf{x}\in \Ball(\mbf{z},\norm{\mbf{z}})| \norm{\mbf{z}}=\rho \} d\mbf{x} \\
&\mathcal{S}_j = \{x: x\geq 0,  e^{\delta}\lambda_{j-1}<x\leq  e^{\delta}\lambda_j\} \\
&\sbf{\mathcal{S}}_j = \{\mbf{x}: \mbf{x}\in\mathds{R}^n, \lambda_{j-1}<\norm{\mbf{x}}\leq \lambda_j\}
\end{align}
where $\alpha^{\textnormal{opt}}(x)$ is as defined in \eqref{eq_alpha1_3}, $\sigma\{\mbf{x}\in \Ball(\mbf{z},\norm{\mbf{z}})\}$ is the characteristic function of $\Ball(\mbf{z},\norm{\mbf{z}})$, $\{\lambda_j\}_{j=0}^\infty$ is the previously defined distance series of $\Lambda_0$, and $M$ is the maximal index $j$ such that $\lambda_j\leq 2r$.
\end{thm}
\begin{proof}
We set $g(\sbf{\lambda})$ as in \eqref{eq_bound_mhs_g} noting that it is bounded by $\lambda_{\textnormal{max}}=2r$.
We continue by remembering that $\alpha^{\textnormal{opt}}(x)$ is piecewise constant in the shells $\mathcal{S}_j$, and therefore \eqref{eq_nrmh2} collapses to
\begin{equation}
\label{eq_nrmh2_alpha1}
\sum_{\sbf{\lambda}\in\Lambda_0/\{0\}} g(\sbf{\lambda})
\leq \beta \sum_{j=0}^M \alpha_j \int_{\sbf{\mathcal{S}}_j} g(\mbf{x}) d\mbf{x}.
\end{equation}
For the remainder we continue in a similar manner to the proof of theorem \ref{thm_bound_mhs} by upper bounding the UBT from \eqref{eq_bound_general} using \eqref{eq_nrmh2_alpha1}. See appendix \ref{app_bound_nrmh2}.
\end{proof}
A geometrical interpretation of $h_j(\rho)$ is presented in figure \ref{fig_h_j}.

\input{h_j.tpx}

\subsection{Sphere Upper Bound (SUB)}
Another bound involving several elements of the spectrum can be constructed by directly bounding the conditional pairwise error probability $P_2(\mbf{x},\rho)$.
This bound is considerably more complex than the DMHS, but is potentially much tighter for the finite dimension case.
When the channel is restricted to AWGN, the resulting bound is similar in performance to the SUB of \cite{j:HerzbergTechniquesofBounding}, hence the name.
The bounding technique for $P_2(\mbf{x},\rho)$, presented in the following lemma, is based on \cite{j:LomnitzCommunicationOver}.
\begin{lem}[Appendix D of \cite{j:LomnitzCommunicationOver}]
\label{lem_prob_ball}
Let $\mbf{x}$ be a vector point in $n$-space with norm $\norm{\mbf{x}}\leq2\rho$, $\mbf{z}$ an isotropically distributed $n$-dimensional random vector, and $\rho$ a real number then
\begin{equation}
\label{eq_prob_ball}
\Pr(\mbf{x}\in \Ball(\mbf{z},\norm{\mbf{z}})|\norm{\mbf{z}}=\rho)
\leq \left( 1- \left( \frac{\norm{\mbf{x}}}{2\rho} \right)^2 \right)^{\frac{n-1}{2}}
\end{equation}
\end{lem}
\begin{proof}
See appendix \ref{app_prob_ball}.
\end{proof}
The above lemma leads directly to the following theorem.
\begin{thm}[SUB]
\label{thm_bound_sub}
Let a specific $n$-dimensional lattice $\Lambda_0$ of density $\beta$ be transmitted over an ACSN channel with $f_{\norm{\mbf{z}}}(\rho)$ the \pdf of its noise vector's norm, then the error probability of an ML decoder is upper bounded by
\begin{align}
\label{eq_bound_sub}
P_e(\Lambda_0)
\leq \min_r & \left( \sum_{j=1}^M \mathcal{N}_j \int_{\lambda_j/2}^{r} f_{\norm{\mbf{z}}}(\rho) \left( 1- \left( \frac{\lambda_j}{2\rho} \right)^2 \right)^{\frac{n-1}{2}} d\rho \right.  \nonumber \\
& \qquad + \left. \int_{r}^\infty f_{\norm{\mbf{z}}}(\rho) d\rho \right)
\end{align}
where $\{\lambda_j\}_{j=1}^\infty$ and $\{\mathcal{N}_j\}_{j=1}^\infty$ are the previously defined distance series and spectrum of $\Lambda_0$ respectively, and $M$ is the maximal index $j$ such that $\lambda_j\leq 2r$.
\end{thm}
\begin{proof}
Bound the UBT of \eqref{eq_bound_general} directly, using \eqref{eq_prob_ball}. See appendix \ref{app_bound_sub}.
\end{proof}

Optimization of $r$ can be performed in the following manner.
We begin by analyzing the function $f(\rho)=\left( 1- \left( \frac{\lambda_j}{2\rho} \right)^2 \right)^{\frac{n-1}{2}}$.
It is simple to verify that this function is positive and strictly increasing in the domain $\{\rho:\rho\geq\lambda_j/2\}$.
Since the UBT of \eqref{eq_bound_sub} is a positive sum of such functions, it is positive and monotonically nondecreasing in $r$.
Since additionally the SBT of \eqref{eq_bound_sub} is always positive, it suffices to search for the optimal $r$ between $\lambda_1/2$ and an $r_{\textnormal{max}}$, where $r_{\textnormal{max}}$ is defined such that the UBT is greater than or equal to $1$.

We continue by calculating $M_{\textnormal{max}}$ that corresponds to $r_{\textnormal{max}}$. By definition, each selection of $M$ corresponds to a domain $\{r:\lambda_M\leq 2r<\lambda_{M+1}\}$.
Instead of searching for the optimal $r$ over the whole domain $\{r:\lambda_1/2\leq r<r_{\textnormal{max}}\}$, we search over all sub-domains corresponding to $1\leq M\leq M_{\textnormal{max}}$.
When $M$ is constant, independent of $r$, the optimal $r$ is found by equating the differential of \eqref{eq_bound_sub} to $0$, in the domain $\{r:\lambda_M\leq 2r<\lambda_{M+1}\}$. Equating the differential to $0$ results in the following condition on $r$:
\begin{equation}
\sum_{j=1}^M \mathcal{N}_j  \left( 1- \left( \frac{\lambda_j}{2r} \right)^2 \right)^{\frac{n-1}{2}}-1 = 0.
\end{equation}
The function on the left of the above condition is monotonically nondecreasing (as shown previously), so an optimal $r$ exists in $\{r:\lambda_M\leq 2r<\lambda_{M+1}\}$ iff its values on the domain edges are of opposite signs. If no optimum exists in the domain, it could exist on one of the domain edges.

The optimization algorithm proceeds as follows: set $M=1$ and examine the differential at the domain edges $\lambda_M/2$ and $\lambda_{M+1}/2$.
If the edges are of opposite signs, find the exact $r$ that zeros the differential in the domain and store it as $r_{\textnormal{zc}}^M$.
Otherwise set $r_{\textnormal{zc}}^M=\emptyset$, advance $M$ by $1$ and repeat the process.
The process is terminated at $M=M_{\textnormal{max}}$.
When done, evaluate \eqref{eq_bound_sub} at $r=r_{\textnormal{zc}}^1,\dots,r_{\textnormal{zc}}^M,\lambda_1/2,\dots,\lambda_{M_{\textnormal{max}}}/2$ and select the $r$ that minimizes.

We conclude this section with an example presented in figure \ref{fig_leech}, which illustrates the effectiveness of the new bounds for finite dimension.
The error probability of the Leech\footnotemark{} lattice $\Lambda_{24}$ is upper bounded by DMHS \eqref{eq_bound_nrmh1} and SUB \eqref{eq_bound_sub}.
The ordinary Union Bound (UB), the MHS bound for dimension $24$ \eqref{eq_bound_mhs}, and the Sphere Lower Bound (SLB) of \cite{j:TarokhUniversalBound} are added for reference.
The spectral data is taken from \cite{b:ConwaySpherePackings}.
The bounds are calculated for an AWGN channel with noise variance $\sigma^2$.
\footnotetext{The Leech lattice is the densest lattice packing in 24 dimensions.}

\begin{figure}[htp]
\center{\includegraphics[width=0.5\textwidth]{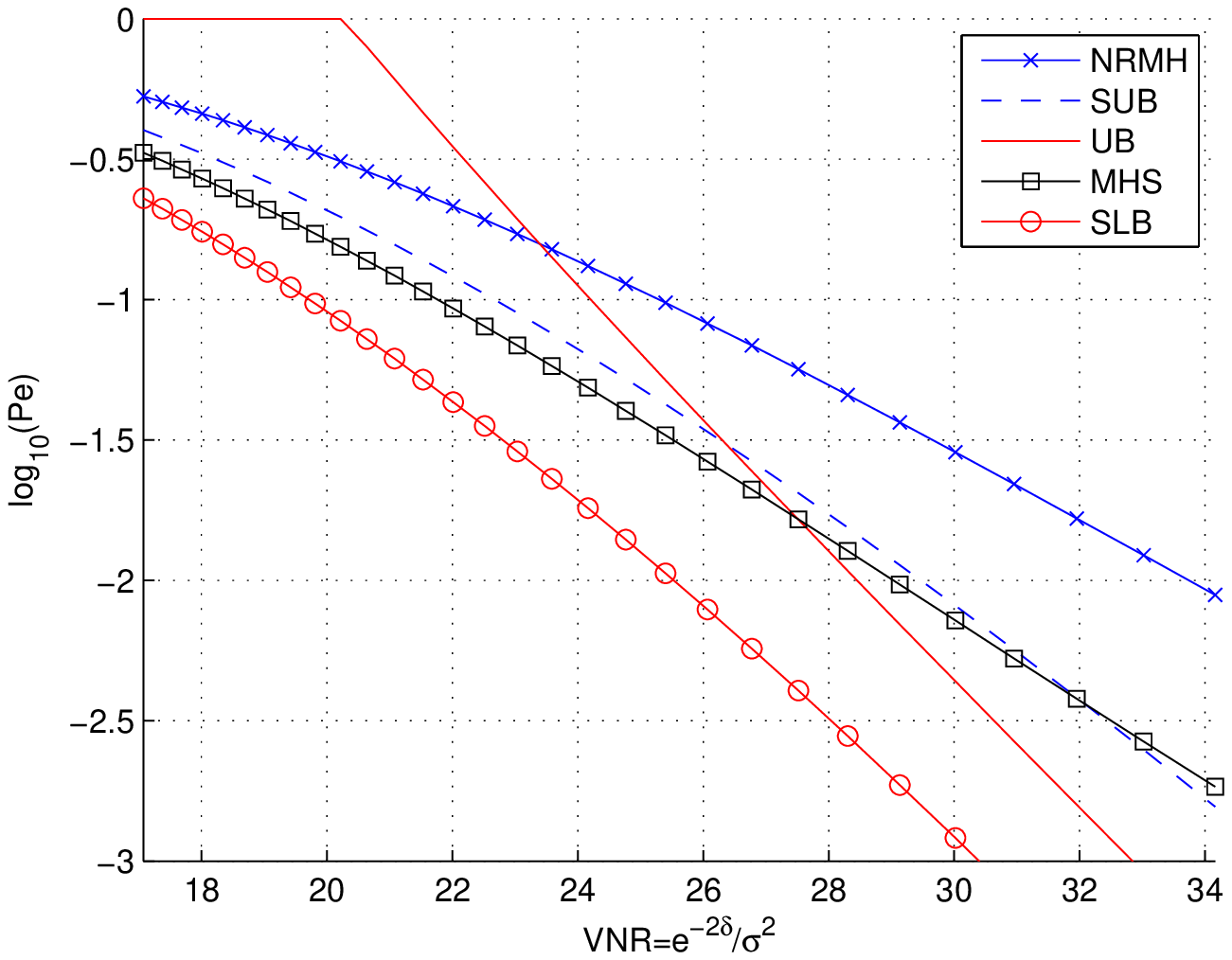}}
\caption{\label{fig_leech} A comparison of DMHS and SUB for the Leech lattice. The UB, MHS and SLB are added for reference. The graph shows the error probability as a function of the Volume-to-Noise Ratio (VNR) for rates $\delta^*<\delta<\delta_{cr}$.}
\end{figure}

\section{Error Exponents for the AWGN Channel}
The previous section presented three new spectrum based upper bounds for specific lattices.
As stated there, only the DMHS bound seems suitable for asymptotical analysis.
This section uses the DMHS bound to construct an error-exponent for specific lattices over the AWGN channel.
Although this case is of prominent importance, it is important to keep in mind that it is only one extension of the DMHS bound.
In general, DMHS bound extensions are applicable wherever MHS bound extensions are.

When the channel is AWGN with noise variance $\sigma^2$, the upper error bound on the ML decoding of a ``good''\footnotemark{} lattice from the MHS ensemble \eqref{eq_bound_mhs} can be expressed in the following exponential form \cite{j:PoltyrevOnCoding}, \cite{j:IngberFiniteDimensional}
\footnotetext{Define a ``good'' lattice, from an ensemble, as one that is upper bounded by the ensemble's average.}
\begin{equation}
\label{eq_bound_mhs_exp}
P_e(\Lambda)
\leq e^{-n(E_r(\delta)+o(1))}
\end{equation}
with
\begin{equation}
E_r(\delta)
= \left\{
\begin{array}{ll}
(\delta^*-\delta) + \log{\frac{e}{4}},                                         & \delta\leq\delta_{cr} \\
\frac{e^{2(\delta^*-\delta)}-2(\delta^*-\delta)-1}{2},
& \delta_{cr}\leq\delta<\delta^* \\
0,                                                                             & \delta\geq\delta^*
\end{array} \right.
\end{equation}
\begin{align}
\delta^* &= \frac{1}{2}\log{\frac{1}{2\pi e\sigma^2}} \\
\delta_{cr} &= \frac{1}{2}\log{\frac{1}{4\pi e\sigma^2}}
\end{align}
where $o(1)$ goes to zero asymptotically with $n$.

This error-exponent can be directly deduced from the MHS bound \eqref{eq_bound_mhs}.
By applying similar methods to the DMHS bound \eqref{eq_bound_nrmh1}, it is possible to construct an error-exponent for a specific lattice \textit{sequence} based on its distance-spectrum.
\begin{thm}[Non-Random Coding Error Exponent]
Let $\Lambda_0[n]$ be a specific lattice \textit{sequence} transmitted over an AWGN channel with noise variance $\sigma^2$, then the error probability of an ML decoder is upper bounded by
\begin{equation}
\label{eq_bound_nrmh1_exp}
P_e(\Lambda_0[n])
\leq e^{-n(E_r(\delta+\nu[n])+o(1))}
\end{equation}
with
\begin{equation}
\nu[n] \triangleq \frac{1}{n}\log{\alpha[n]}.
\end{equation}
\end{thm}
where $[n]$ indicates the $n$'th element of the sequence.
\begin{proof}
It follows from the proof of theorem \ref{thm_bound_nrmh1}, that replacing $\beta$ with $\beta\alpha[n]$ there, is equivalent to replacing $\delta$ with $\delta+\nu[n]$ here.
\end{proof}
Clearly \eqref{eq_bound_nrmh1_exp} can be used to determine the exponential decay of the error probability of a specific lattice \textit{sequence}. This leads to the following corollary.
\begin{cor}[Gap to Capacity 1]
A lattice \textit{sequence} for which
\begin{equation}
\lim_{n\rightarrow\infty} \frac{1}{n}\log{\alpha[n]} = 0
\end{equation}
achieves the unrestricted channel error-exponent.
\end{cor}
\begin{proof}
Follows immediately from \eqref{eq_bound_nrmh1_exp} and the definition of $\nu[n]$.
\end{proof}
When certain stronger conditions on the lattice apply, the following simpler corollary can be used.
\begin{cor}[Gap to Capacity 2]
\label{cor_cap_gap_2}
A lattice \textit{sequence} for which $\alpha^{\textnormal{rng}}(x)$ (per dimension $n$) is monotonically non-increasing in $x$ and
\begin{equation}
\label{eq_cor_cap_gap_2}
\lim_{n\rightarrow\infty} \frac{1}{n} \log \left( \frac{\mathcal{N}_1[n]}{e^{n\delta}V_n(\lambda_1[n])^n} \right) = 0
\end{equation}
achieves the unrestricted channel error-exponent.
\end{cor}
\begin{proof}
When $\alpha^{\textnormal{rng}}(x)$ is monotonically non-increasing then $\alpha^{\textnormal{opt}}(x) = \alpha^{\textnormal{rng}}(x)$ and $\alpha[n] = \frac{\mathcal{N}_1[n]}{e^{n\delta}V_n(\lambda_1[n])^n}$.
\end{proof}
Let us examine corollary \ref{cor_cap_gap_2}.
Assume a sequence of lattices with minimum distance $\lambda_1[n]=e^{-\delta}V_n^{-1/n}$ is available (see \cite{p:IngberExpurgatedInfiniteConstellations} for proof of existence).
Plugging this back into \eqref{eq_cor_cap_gap_2} leads to the following necessary condition for the bound \eqref{eq_bound_nrmh1_exp} to achieve the unrestricted channel error-exponent
\begin{equation}
\label{eq_cor_cap_gap_2_eg}
\lim_{n\rightarrow\infty} \frac{1}{n} \log \left( \mathcal{N}_1[n] \right) = 0,
\end{equation}
whether the monotonicity condition applies or not.
Although this only amounts to conditioning on the bound (and not on the lattice sequence), condition \eqref{eq_cor_cap_gap_2_eg} gives an insight to the close relationship between the spectral distances and their enumeration.
We conjecture that at rates close to capacity, the bound is tight, leading to \eqref{eq_cor_cap_gap_2} being a necessary condition on the lattice sequence itself.

We conclude this section with an illustration of the exponential decay series $\nu[n]$ of the first three lattices of the Barnes-Wall lattice \textit{sequence} $BW_4=D_4$, $BW_8=E_8$, and $BW_{16}=\Lambda_{16}$.
Unfortunately the distance-spectrum for $BW_n$ is generally unknown, preventing asymptotical analysis.
Nonetheless an interpolation for dimension $4$ to $16$ is presented in figure \ref{fig_bw_nu}.
The spectral data is taken from \cite{b:ConwaySpherePackings}.

\begin{figure}[htp]
\center{\includegraphics[width=0.5\textwidth]{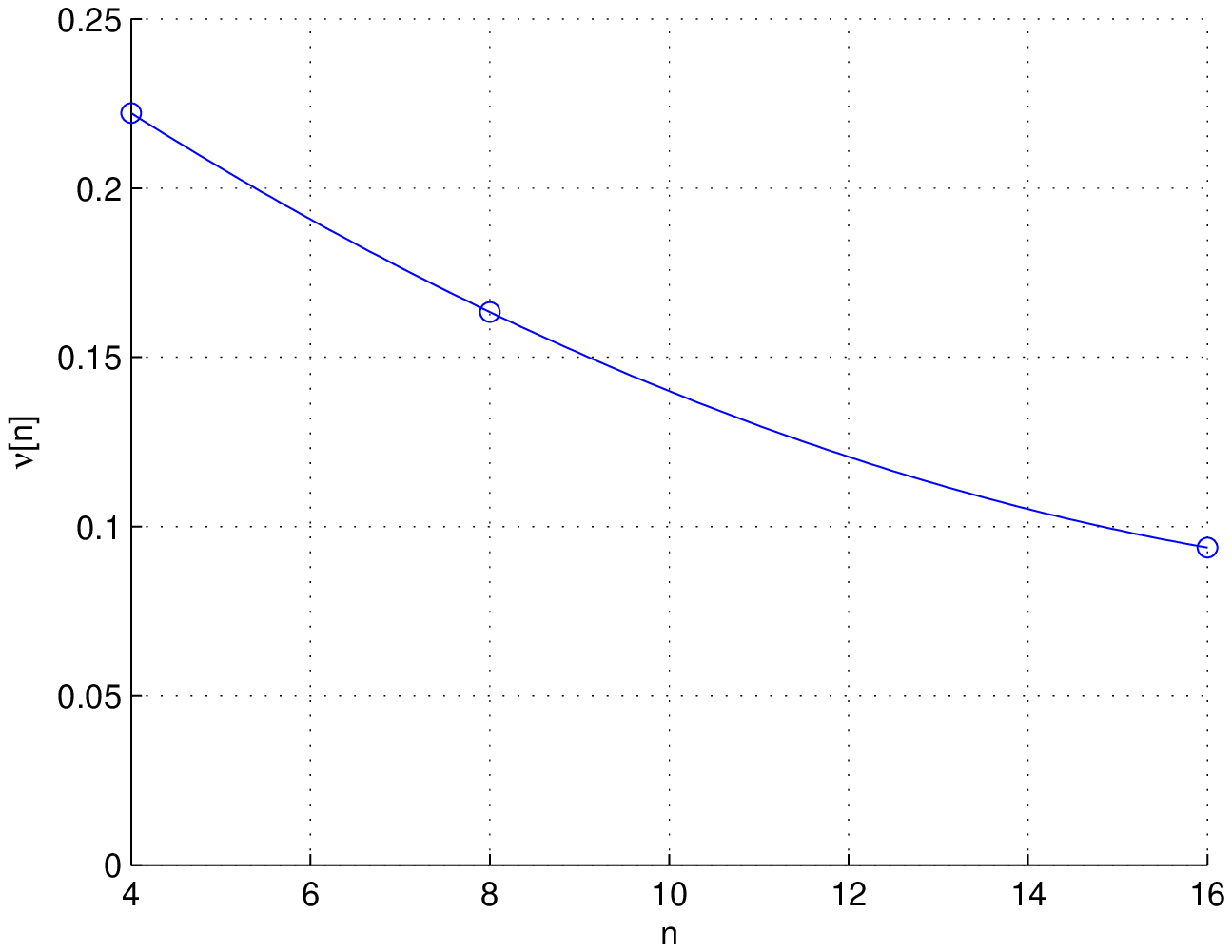}}
\caption{\label{fig_bw_nu} The exponential decay series $\nu[n]$ for the lattice \textit{sequence} $BW_n$, calculated for dimensions $4$, $8$, and $16$ and interpolated in-between.}
\end{figure}

Examination of figure \ref{fig_bw_nu} shows that the upper bound on the gap to capacity decreases with the increase in $n$, at least for the first three lattices in the sequence.
Although full spectral data is not available for the remainder of the sequence, the minimal distance and its enumeration are known analytically.
Assuming momentarily that the condition on $\alpha^{\textnormal{rng}}(x)$ as presented in corollary \ref{cor_cap_gap_2} holds, we can try to examine $\frac{1}{n} \log \left( \frac{\mathcal{N}_1[n]}{e^{n\delta}V_n(\lambda_1[n])^n} \right)$ as an upper bound on the gap to capacity.
This examination is illustrated in figure \ref{fig_bw_nu_ext}, where the first three lattices coincide with the previous results in figure \ref{fig_bw_nu}.
Clearly, these results are a lower bound on an upper bound and cannot indicate one way or the other.
Nonetheless, it seems that the results are consistent with the well-known coding performance of Barnes-Wall lattices with increasing dimension.

\begin{figure}[htp]
\center{\includegraphics[width=0.5\textwidth]{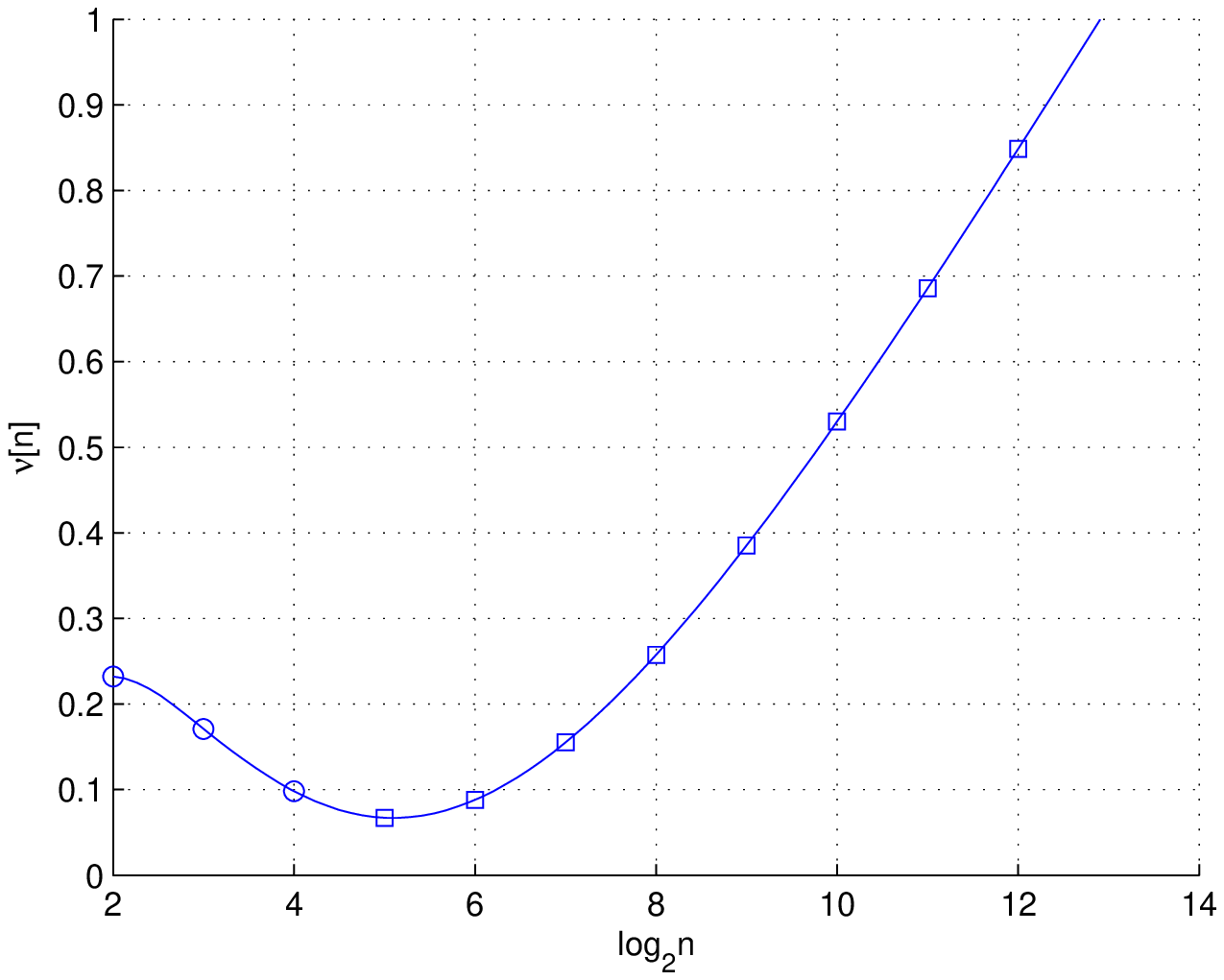}}
\caption{\label{fig_bw_nu_ext} $\frac{1}{n} \log \left( \frac{\mathcal{N}_1[n]}{e^{n\delta}V_n(\lambda_1[n])^n} \right)$ for the first few dimensions of the lattice \textit{sequence} $BW_n$. Coincides with $\frac{1}{n}\log{\alpha[n]}$ for the first three lattices.}
\end{figure}

\appendix
\subsection{Linear Program for Finding $\alpha^{\textnormal{opt}}(x)$}
\label{app_linear_program_alpha1}
We begin by restating the conditions describing $\alpha^{\textnormal{opt}}(x)$.
\begin{align}
\label{eq_linear_program_alpha1_1}
&\alpha^{\textnormal{opt}}(x) = \alpha_j(x) \nonumber \\
&j = \argmin_{i \in I} \max_{x\leq  e^{\delta}\lambda_{\textnormal{max}}} \alpha_i(x)
\end{align}
where $\alpha(x)$ is expressed by
\begin{equation}
\label{eq_linear_program_alpha1_2}
\alpha(x)
\triangleq \left\{
\begin{array}{lr}
\frac{N'(x)}{nV_n x^{n-1}} & : x \neq 0 \\
0 & : x = 0
\end{array}
\right.
\end{equation}
and $N'(x)$ satisfies the inequality
\begin{equation}
\label{eq_linear_program_alpha1_3}
\int_{0^+}^r N(x) dx
\leq \int_{0^+}^r N'(x) dx
\qquad \forall r\in (0,e^{\delta}\lambda_{\textnormal{max}}].
\end{equation}
Recall that $N'(x)$ is a smoothed version of the continuous spectrum $N(x)$ such that spectral mass is allowed to shift down from higher to lower distances.
As is apparent from the conditions above, finding $\alpha^{\textnormal{opt}}(x)$ is equivalent to finding $N'(x)$ such that \eqref{eq_linear_program_alpha1_1} is optimized.
It is immediately clear from \eqref{eq_linear_program_alpha1_1} that since we are searching for the function that minimizes its maximum value over the domain $\{x: 0<x\leq e^{\delta}\lambda_{\textnormal{max}}\}$, we can only benefit from limiting $N'(x)$ to consist of only the spectral elements $\{\mathcal{N}_j\}_{j=1}^M$ where $M$ is the maximal index such that $\lambda_j\leq\lambda_{\textnormal{max}}$.
This limiting necessarily does not increase the max value over the domain $\{x: 0<x\leq  e^{\delta}\lambda_{\textnormal{max}}\}$ and is consistent with \eqref{eq_linear_program_alpha1_3} for $r=\lambda_{\textnormal{max}}$.
Clearly, optimization of \eqref{eq_linear_program_alpha1_1} is achieved by an $\alpha(x)$ that is as constant as possible over the domain $\{x: 0<x\leq e^{\delta}\lambda_{\textnormal{max}}\}$, as is typical for a $\min\max$ optimization.
We say ``as constant as possible'' since \eqref{eq_linear_program_alpha1_3} imposes constraints on our ability to find $N'(x)$ such that $\alpha(x)$ is constant over the whole domain $\{x: 0<x\leq  e^{\delta}\lambda_{\textnormal{max}}\}$. Given \eqref{eq_linear_program_alpha1_3}, the optimum $\alpha(x)$ is piecewise constant over the sub-domains $\{x: \widetilde{\lambda}_{j-1}<x\leq\widetilde{\lambda}_j\}$ with $1\leq j\leq M$, and is zero in the sub-domain $\{x: \widetilde{\lambda}_M<x\leq\widetilde{\lambda}_{\textnormal{max}}\}$. With the piecewise nature of $\alpha(x)$ established, and due to \eqref{eq_linear_program_alpha1_2}
\begin{align}
\label{eq_linear_program_alpha1_4}
\int_{\widetilde{\lambda}_{j-1}}^{\widetilde{\lambda}_j}N'(x)
&= \int_{\widetilde{\lambda}_{j-1}}^{\widetilde{\lambda}_j}\alpha(x)nV_n x^{n-1}dx \nonumber \\
&= \alpha_j\int_{\widetilde{\lambda}_{j-1}}^{\widetilde{\lambda}_j}nV_n x^{n-1}dx \nonumber  \\
&= \alpha_j V_n\left(\widetilde{\lambda}_{j}^n-\widetilde{\lambda}_{j-1}^n\right)
\end{align}
where $\alpha_j$ is the constant value of $\alpha(x)$ over sub-domain $\{x: \widetilde{\lambda}_{j-1}<x\leq\widetilde{\lambda}_j\}$.
It is quite clear from \eqref{eq_linear_program_alpha1_4}, that in our case $N'(x)$ can be completely described by $\epsilon_j = \int_{\widetilde{\lambda}_{j-1}}^{\widetilde{\lambda}_j}N'(x)$ for $1\leq j\leq M$.

We now present the general formulation necessary to construct the LP for optimizing $\alpha(x)$.
Take each spectral element $\mathcal{N}_j$ and subdivide it into $j$ parts $\{\mathcal{N}_{j,i}\}_{i=1}^j$ such that $\mathcal{N}_{j,i}\geq0$ and $\mathcal{N}_j = \sum_{i=1}^j \mathcal{N}_{j,i}$.
$\mathcal{N}_{j,i}$ can be viewed as the contribution of the $j$'th spectral element $\mathcal{N}_j$ to the $i$'th sub-domain $\{x: \widetilde{\lambda}_{i-1}<x\leq\widetilde{\lambda}_i\}$ where $1\leq i\leq j$.
Defining $\epsilon_i$ to be the total spectral mass associated with sub-domain $\{x: \widetilde{\lambda}_{i-1}<x\leq\widetilde{\lambda}_i\}$, we can summarize the description of the LP variables $\mathcal{N}_{j,i}$ as follows:
\begin{align}
&\mathcal{N}_{j,i} \geq0                        &\qquad \forall \quad i\leq j\leq M \\
&\sum_{i=1}^j \mathcal{N}_{j,i} = \mathcal{N}_j &\qquad \forall \quad 1\leq j\leq M \\
&\sum_{j=i}^M \mathcal{N}_{j,i} = \epsilon_i    &\qquad \forall \quad 1\leq j\leq M
\end{align}
An illustration of the LP variables for $M=4$ is presented in the following table:
\begin{table}[h]
\centering
\begin{tabular}{|c||c|c|c|c|}
\hline
             & $\mathcal{N}_1$     & $\mathcal{N}_2$     & $\mathcal{N}_3$     & $\mathcal{N}_4$ \\
\hline\hline
$\epsilon_1$ & $\mathcal{N}_{1,1}$ & $\mathcal{N}_{2,1}$ & $\mathcal{N}_{3,1}$ & $\mathcal{N}_{4,1}$ \\
\hline
$\epsilon_2$ & -                   & $\mathcal{N}_{2,2}$ & $\mathcal{N}_{3,2}$ & $\mathcal{N}_{4,2}$ \\
\hline
$\epsilon_3$ & -                   & -                   & $\mathcal{N}_{3,3}$ & $\mathcal{N}_{4,3}$ \\
\hline
$\epsilon_4$ & -                   & -                   & -                   & $\mathcal{N}_{4,4}$ \\
\hline
\end{tabular}
\end{table}

The LP variables set $\{\{\mathcal{N}_{j,i}\}_{i=1}^j\}_{j=1}^M$ completely defines $N'(x)$, such that its optimum completely defines $\alpha^{\textnormal{opt}}(x)$.
Finally the LP is defined as follows:
\begin{align}
\label{eq_linear_program_alpha1_5}
&\text{minimize } \alpha : \nonumber \\
&\mathcal{N}_{j,i} \geq0                                                                                                  &\qquad \forall \quad i\leq j\leq M \\
&\sum_{i=1}^j \mathcal{N}_{j,i} = \mathcal{N}_j                                                                           &\qquad \forall \quad 1\leq j\leq M \\
&\frac{\sum_{j=i}^M \mathcal{N}_{j,i}}{V_n\left(\widetilde{\lambda}_{j}^n-\widetilde{\lambda}_{j-1}^n\right)} \leq \alpha &\qquad \forall \quad 1\leq j\leq M.
\end{align}
It is easy to show that the LP optimization landscape is bounded, by bounding each variable individually
\begin{equation}
0 \leq \mathcal{N}_{j,i} \leq \mathcal{N}_j \qquad \forall \quad i\leq j\leq M.
\end{equation}
It is simple to show that there exists a solution to the LP constraints. One trivial solution is
\begin{equation}
\mathcal{N}_{j,i}
= \left\{
\begin{array}{lr}
\mathcal{N}_j & : i=j \\
0 & : i\neq j
\end{array}
\right.
\qquad \forall \quad i\leq j\leq M.
\end{equation}
The LP can be solved by a water-filling algorithm which is linear in the number of considered spectral elements $M$.
At the $j$'th step the water (a.k.a. spectral mass $\mathcal{N}_j$) is poured into the $j$'th bucket (a.k.a. sub-domain $\{x: \widetilde{\lambda}_{j-1}<x\leq\widetilde{\lambda}_j\}$), until it overflows (a.k.a. exceeds the previous maximum), at which point it is permitted to pour over lower sub-domains.
A graphical representation of the algorithm is presented in figure \ref{fig_lp_wf}.
\input{lp_wf.tpx}
The optimality of the water-filling solution can be verified by induction.

\subsection{Proof of Theorem \ref{thm_bound_general}}
\label{app_bound_general}
We use Herzberg and Poltyrev's upper sphere bound definition for AWGN \cite{j:HerzbergTechniquesofBounding} to introduce the general technique in the context of lattices over ACSN channels.
Consider an $n$-dimensional lattice $\Lambda$ transmitted over an ACSN channel. In each transmission a lattice point $\sbf{\lambda}$ is transmitted and $\mbf{y}=\sbf{\lambda}+\mbf{z}$ is received. The additive noise $\mbf{z}$ corrupts the transmitted lattice point. An ML decoder decodes correctly so long as $\mbf{z}$ is contained in the voronoi cell surrounding $\sbf{\lambda}$.
The probability of decoding error can be expressed as
\begin{align}
\label{eq_gallager_1}
P_e(\Lambda)
= \Pr(e|\norm{\mbf{z}}\leq r) \Pr(\norm{\mbf{z}}\leq r) \nonumber \\
+ \Pr(e|\norm{\mbf{z}}>r) \Pr(\norm{\mbf{z}}>r)
\end{align}
where $e$ denotes a decoding error event, and $r$ is a real positive parameter.
Since $\Pr(\cdot)\leq 1$, equation \eqref{eq_gallager_1} can be trivially upper bounded as
\begin{equation}
\label{eq_gallager_2}
P_e(\Lambda)
\leq \Pr(e|\norm{\mbf{z}}\leq r) + \Pr(\norm{\mbf{z}}>r).
\end{equation}
Assuming without loss of generality that $\sbf{\lambda}=0$ was transmitted, the first term of \eqref{eq_gallager_2} can be upper bounded as
\begin{align}
\label{eq_herzberg_1}
\Pr(e|\norm{\mbf{z}}\leq r)
&= \Pr \left( \bigcup_{\sbf{\lambda}\in\Lambda/\{0\}} (e_{\sbf{\lambda}}|\norm{\mbf{z}}\leq r) \right) \nonumber \\
&\leq \sum_{\sbf{\lambda}\in\Lambda/\{0\}} \Pr(e_{\sbf{\lambda}}|\norm{\mbf{z}}\leq r)
\end{align}
where $e_{\sbf{\lambda}}$ is a pairwize error resulting from decoding $\sbf{\lambda}\neq 0$, and the inequality is due to the union bound.
The bound after substitution of \eqref{eq_herzberg_1} in \eqref{eq_gallager_2} is given by
\begin{equation}
\label{eq_herzberg_2}
P_e(\Lambda)
\leq \sum_{\sbf{\lambda}\in\Lambda/\{0\}} \Pr(e_{\sbf{\lambda}}|\norm{\mbf{z}}\leq r) + \Pr(\norm{\mbf{z}}>r).
\end{equation}
Consequently, the tightest bound is found by optimization over $r$,
\begin{equation}
\label{eq_herzberg_3}
P_e(\Lambda) \leq
\min_{r} \left(
\sum_{\sbf{\lambda}\in\Lambda/\{0\}} \Pr(e_{\sbf{\lambda}}|\norm{\mbf{z}}\leq r) + \Pr(\norm{\mbf{z}}>r)
\right).
\end{equation}
Following techniques from Ingber \etal \cite{j:IngberFiniteDimensional}, the inner probability term of the first term on the right-hand-side of \eqref{eq_herzberg_2} can be expressed as\footnotemark{}
\begin{align}
\label{eq_ingber_1}
\Pr(e_{\sbf{\lambda}}|&\norm{\mbf{z}}\leq r) \nonumber \\
&= \int_0^r f_{\norm{\mbf{z}}}(\rho) \Pr(e_{\sbf{\lambda}}|\norm{\mbf{z}}=\rho) d\rho \nonumber \\
&= \int_0^r f_{\norm{\mbf{z}}}(\rho) \Pr(\norm{\mbf{z}-\sbf{\lambda}}\leq\norm{\mbf{z}}|\norm{\mbf{z}}=\rho) d\rho \nonumber \\
&= \int_0^r f_{\norm{\mbf{z}}}(\rho) \Pr(\sbf{\lambda}\in \Ball(\mbf{z},\norm{\mbf{z}})|\norm{\mbf{z}}=\rho) d\rho
\end{align}
where $f_{\norm{\mbf{z}}}(\rho)$ is the \pdf of the noise vector's norm, and $\Ball(\mbf{z},\norm{\mbf{z}})$ is an $n$-dimensional ball of radius $\norm{\mbf{z}}$ centered around $\mbf{z}$.
\footnotetext{The integral in \eqref{eq_ingber_1} can be further constrained to $[\norm{\sbf{\lambda}}/2,r]$ since $\Pr(\sbf{\lambda}\in \Ball(\mbf{z},\norm{\mbf{z}})|\norm{\mbf{z}}=\rho) = 0$ for $\rho<\norm{\sbf{\lambda}}/2$.}
Using the same notations, the second term on the right-hand-side of \eqref{eq_herzberg_2} can be expressed as
\begin{equation}
\label{eq_ingber_2}
\Pr(\norm{\mbf{z}}>r) = \int_r^\infty f_{\norm{\mbf{z}}}(\rho) d\rho.
\end{equation}
Plugging \eqref{eq_ingber_1} and \eqref{eq_ingber_2} into \eqref{eq_herzberg_3} completes the proof.
\subsection{Proof of Theorem \ref{thm_bound_mhs}}
\label{app_bound_mhs}
This proof is brought for completeness and reading convenience, although it has been seen, in similar form, elsewhere, previously.

Using inequality \eqref{eq_mhs},
\begin{align}
& \sum_{\sbf{\lambda}\in\Lambda/\{0\}} \int_0^r f_{\norm{\mbf{z}}}(\rho) \Pr(\sbf{\lambda}\in \Ball(\mbf{z},\norm{\mbf{z}})|\norm{\mbf{z}}=\rho) d\rho \nonumber \\
& \leq \beta \int_{\mathds{R}^n} \int_0^r f_{\norm{\mbf{z}}}(\rho) \Pr(\mbf{x}\in \Ball(\mbf{z},\norm{\mbf{z}})|\norm{\mbf{z}}=\rho) d\rho d\mbf{x} \nonumber \\
& = \beta \int_0^r f_{\norm{\mbf{z}}}(\rho) \int_{\mathds{R}^n} \Pr(\mbf{x}\in \Ball(\mbf{z},\norm{\mbf{z}})|\norm{\mbf{z}}=\rho) d\mbf{x} d\rho \nonumber \\
& = \beta \int_0^r f_{\norm{\mbf{z}}}(\rho) \int_{\mathds{R}^n} \E(\sigma\{\mbf{x}\in \Ball(\mbf{z},\norm{\mbf{z}})\}|\norm{\mbf{z}}=\rho) d\mbf{x} d\rho \nonumber \\
& = \beta \int_0^r f_{\norm{\mbf{z}}}(\rho) \E \left( \left. \int_{\mathds{R}^n} \sigma\{\mbf{x}\in \Ball(\mbf{z},\norm{\mbf{z}})\} d\mbf{x} \right| \norm{\mbf{z}}=\rho \right) d\rho \nonumber \\
& = \beta \int_0^r f_{\norm{\mbf{z}}}(\rho) \int_{\mathds{R}^n} \sigma\{\mbf{x}\in \Ball(\mbf{z},\norm{\mbf{z}})| \norm{\mbf{z}}=\rho \} d\mbf{x} d\rho \nonumber \\
& = \beta V_n \int_0^r f_{\norm{\mbf{z}}}(\rho) \rho^n d\rho
\end{align}
where $\sigma\{\mbf{x}\in \Ball(\mbf{z},\norm{\mbf{z}})\}$ is the characteristic function of $\Ball(\mbf{z},\norm{\mbf{z}})$. The optimal value $r^*$ of $r$ is found by differentiating \eqref{eq_bound_mhs},
\begin{equation}
\beta V_n f_{\norm{\mbf{z}}}(r) r^n = f_{\norm{\mbf{z}}}(r).
\end{equation}
\subsection{Proof of Theorem \ref{thm_bound_nrmh2}}
\label{app_bound_nrmh2}
Using inequality \eqref{eq_nrmh2_alpha1},
\begin{align}
\label{eq_bound_nrmh2_proof_1}
& \sum_{\sbf{\lambda}\in\Lambda_0/\{0\}} \int_0^r f_{\norm{\mbf{z}}}(\rho) \Pr(\sbf{\lambda}\in \Ball(\mbf{z},\norm{\mbf{z}})|\norm{\mbf{z}}=\rho) d\rho \nonumber \\
& \leq \beta \sum_{j=0}^M \alpha_j \int_{\mathcal{S}_j} \int_0^r f_{\norm{\mbf{z}}}(\rho) \Pr(\mbf{x}\in \Ball(\mbf{z},\norm{\mbf{z}})|\norm{\mbf{z}}=\rho) d\rho d\mbf{x} \nonumber \\
\end{align}
Continue with the term inside the sum \eqref{eq_bound_nrmh2_proof_1},
\begin{align}
\label{eq_bound_nrmh2_proof_2}
& \int_{\mathcal{S}_j} \int_0^r f_{\norm{\mbf{z}}}(\rho) \Pr(\mbf{x}\in \Ball(\mbf{z},\norm{\mbf{z}})|\norm{\mbf{z}}=\rho) d\rho d\mbf{x} \nonumber \\
& = \int_0^r f_{\norm{\mbf{z}}}(\rho) \int_{\mathcal{S}_j} \Pr(\mbf{x}\in \Ball(\mbf{z},\norm{\mbf{z}})|\norm{\mbf{z}}=\rho) d\mbf{x} d\rho \nonumber \\
& = \int_0^r f_{\norm{\mbf{z}}}(\rho) \int_{\mathcal{S}_j} \E(\sigma\{\mbf{x}\in \Ball(\mbf{z},\norm{\mbf{z}})\}|\norm{\mbf{z}}=\rho) d\mbf{x} d\rho \nonumber \\
& = \int_0^r f_{\norm{\mbf{z}}}(\rho) \E \left( \left. \int_{\mathcal{S}_j} \sigma\{\mbf{x}\in \Ball(\mbf{z},\norm{\mbf{z}})\} d\mbf{x} \right| \norm{\mbf{z}}=\rho \right) d\rho \nonumber \\
& = \int_0^r f_{\norm{\mbf{z}}}(\rho) \int_{\mathcal{S}_j} \sigma\{\mbf{x}\in \Ball(\mbf{z},\norm{\mbf{z}})| \norm{\mbf{z}}=\rho \} d\mbf{x} d\rho \nonumber \\
& = \int_0^r f_{\norm{\mbf{z}}}(\rho) h_j(\rho) d\rho \nonumber \\
& = \int_{\lambda_j/2}^r f_{\norm{\mbf{z}}}(\rho) h_j(\rho) d\rho
\end{align}
where $\sigma\{\mbf{x}\in \Ball(\mbf{z},\norm{\mbf{z}})\}$ is the characteristic function of $\Ball(\mbf{z},\norm{\mbf{z}})$, and the last two equalities follow from the definition of $h_j(\rho)$.
Plugging \eqref{eq_bound_nrmh2_proof_2} back into \eqref{eq_bound_nrmh2_proof_1} completes the proof.
\subsection{Proof of Lemma \ref{lem_prob_ball}}
\label{app_prob_ball}
Clearly when $\norm{\mbf{x}}>2\rho$, the probability $\Pr(\mbf{x}\in \Ball(\mbf{z},\norm{\mbf{z}})|\norm{\mbf{z}}=\rho)=0$.

Begin by defining the i.i.d. Gaussian vector
\begin{equation}
\mbf{u} \sim \mathcal{N}^n(0,1).
\end{equation}
The vector $\mbf{z}$ conditioned such that $\norm{\mbf{z}}=\rho$ can be expressed as
\begin{equation}
\mbf{z} = \frac{\rho}{\norm{\mbf{u}}}\mbf{u}
\end{equation}
with no loss of generality.
Since $\mbf{u}$ is isotropically distributed, the vector $\mbf{x}$ can be set to
\begin{equation}
\mbf{x} = [\norm{\mbf{x}},0,\cdots,0]
\end{equation}
again with no loss of generality.
Continue by noting that $\mbf{x}\in \Ball(\mbf{z},\norm{\mbf{z}})$ is equivalent to $\norm{\mbf{z}-\mbf{x}}\leq\rho$, thus
\begin{align}
&\Pr(\mbf{x}\in \Ball(\mbf{z},\norm{\mbf{z}})|\norm{\mbf{z}}=\rho) \nonumber \\
& = \Pr(\norm{\mbf{z}-\mbf{x}}\leq\rho) \nonumber \\
& = \Pr\left(\norm{\frac{\rho}{\norm{\mbf{z}}}[z_1,\cdots,z_n]-[\norm{\mbf{x}},0,\cdots,0]}\leq\rho\right) \nonumber \\
& = \Pr\left(\norm{[z_1-\frac{\norm{\mbf{x}}\norm{\mbf{z}}}{\rho},z_2,\cdots,z_n]}\leq\norm{\mbf{z}}\right) \nonumber \\
& = \Pr\left(\left(z_1-\frac{\norm{\mbf{x}}\norm{\mbf{z}}}{\rho}\right)^2\leq z_1^2\right) \nonumber \\
& = \Pr\left(z_1\geq\frac{\norm{\mbf{x}}\norm{\mbf{z}}}{2\rho}\right) \nonumber \\
& = \frac{1}{2}\Pr\left(z_1^2\geq\frac{(\norm{\mbf{x}}/2\rho)^2}{1-\norm{\mbf{x}}/2\rho)^2}\sum_{i=2}^{n}z_i^2\right).
\end{align}
Substitute $t=(\norm{\mbf{x}}/2\rho)^2$ for easier readability, thus
\begin{align}
&\frac{1}{2}\Pr\left(z_1^2\geq\frac{t}{1-t}\norm{z_{2:n}}^2\right) \nonumber \\
& = \E\left(\frac{1}{2}\Pr\left(\left.z_1^2\geq\frac{t}{1-t}\norm{z_{2:n}}^2\right|z_{2:n}\right)\right) \nonumber \\
& = \E\left(Q\left(\sqrt{\frac{t}{1-t}\norm{z_{2:n}}^2}\right)\right) \nonumber \\
& \leq \E\left(\exp\left(-\frac{1}{2}\frac{t}{1-t}\norm{z_{2:n}}^2\right)\right) \nonumber \\
& = \int_{\mathds{R}^{n-1}} \exp\left(-\frac{1}{2}\frac{t}{1-t}\norm{z_{2:n}}^2\right) f_{\mathcal{N}^{n-1}(0,1)}(z_{2:n})dz_{2:n}\nonumber \\
& = (1-t)^{\frac{n-1}{2}} \int_{\mathds{R}^{n-1}} f_{\mathcal{N}^{n-1}(0,1-t)}(z_{2:n})dz_{2:n}\nonumber \\
& = (1-t)^{\frac{n-1}{2}}
\end{align}
where the inequality results from a fairly crude upper bound on the Gaussian error function $Q(x)\triangleq \Pr(\mathcal{N}(0,1)\geq x)\leq\exp(-x^2/2)$.
Undoing the substitution of $t$ completes the proof.
\subsection{Proof of Theorem \ref{thm_bound_sub}}
\label{app_bound_sub}
\begin{align}
\label{eq_bound_sub_proof}
& \sum_{\sbf{\lambda}\in\Lambda_0/\{0\}} \int_0^r f_{\norm{\mbf{z}}}(\rho) \Pr(\sbf{\lambda}\in \Ball(\mbf{z},\norm{\mbf{z}})|\norm{\mbf{z}}=\rho) d\rho \nonumber \\
&= \sum_{\sbf{\lambda}\in\Lambda_0/\{0\}} \int_{\norm{\sbf{\lambda}}/2}^r f_{\norm{\mbf{z}}}(\rho) \Pr(\sbf{\lambda}\in \Ball(\mbf{z},\norm{\mbf{z}})|\norm{\mbf{z}}=\rho) d\rho \nonumber \\
&\leq \sum_{\sbf{\lambda}\in\Lambda_0/\{0\}} \int_{\norm{\sbf{\lambda}}/2}^r f_{\norm{\mbf{z}}}(\rho) \left( 1- \left( \frac{\norm{\sbf{\lambda}}}{2\rho} \right)^2 \right)^{\frac{n-1}{2}} d\rho \nonumber \\
&= \sum_{j=1}^M \mathcal{N}_j \int_{\lambda_j/2}^r f_{\norm{\mbf{z}}}(\rho) \left( 1- \left( \frac{\lambda_j}{2\rho} \right)^2 \right)^{\frac{n-1}{2}} d\rho
\end{align}
where the first equality is due to $\Pr(\sbf{\lambda}\in \Ball(\mbf{z},\norm{\mbf{z}})|\norm{\mbf{z}}=\rho) = 0$ for $\rho < \norm{\sbf{\lambda}}/2$, the inequality is due to \eqref{eq_prob_ball}, and the last equality is formed by aggregation of identical distance summation terms.

\bibliographystyle{IEEEtran}
\bibliography{IEEEabrv,Yuval}

\end{document}